\documentclass[submission,copyright,creativecommons,hidelinks]{eptcs}
\providecommand{\event}{SCSS 2021} 
\usepackage{breakurl}             
\usepackage{underscore}           
\usepackage{cite}
\usepackage{amssymb}
\usepackage{amsmath}
\usepackage{amsthm}
\usepackage{bussproofs}
\usepackage{multicol}
\usepackage{comment}
\usepackage{enumitem}
\usepackage{url}
\usepackage{xcolor}
\usepackage{tabularx}
\usepackage{framed}
\usepackage{float}
\usepackage[figurename=Figure]{caption}
\usepackage{extarrows}
\usepackage{longtable}

\newtheorem{thm}{Theorem}

\newtheorem{lem}[thm]{Lemma}
\newtheorem{defn}[thm]{Definition}
\newtheorem{cor}[thm]{Corollary}
\newtheorem{exam}{Example}

\DeclareMathAlphabet{\mathrmbf}{\encodingdefault}{\rmdefault}{bx}{n}
\title{Congruence Closure Modulo Permutation Equations}

\author{Dohan Kim and Christopher Lynch
\institute{Clarkson University, Potsdam, NY, USA}
\email{\{dohkim,clynch\}@clarkson.edu}
}

\def\titlerunning{Congruence Closure Modulo Permutation Equations}
\def\authorrunning{D. Kim and C. Lynch}

\begin{document}
\maketitle

\begin{abstract}
We present a framework for constructing congruence closure modulo permutation equations, which extends the \emph{abstract congruence closure}~\cite{Tiwari2003} framework for handling permutation function symbols. Our framework also handles certain interpreted function symbols satisfying each of the following properties: idempotency ($I$), nilpotency ($N$), unit ($U$), $I\cup U$, or $N\cup U$. Moreover, it yields convergent rewrite systems corresponding to ground equations containing permutation function symbols. We show that congruence closure modulo a given finite set of permutation equations can be constructed in polynomial time using equational inference rules, allowing us to provide a polynomial time decision procedure for the word problem for a finite set of ground equations with a fixed set of permutation function symbols.
\end{abstract}

\section{Introduction}
Congruence closure procedures~\cite{Downey1980,Nelson1980,Kozen1977} have been researched for several decades, and play important roles in software/hardware verification (see~\cite{Nelson1980,Cyrluk1995,Sjoberg2015}) and satisfiability modulo theories (SMT) solvers~\cite{Barrett2018, Moura2011}. They provide fast decision procedures for determining whether a ground equation is an (equational) consequence of a given set of ground equations. (The fastest known congruence closure algorithm runs in $O(n\,\text{log}\,n)$~\cite{Kapur2019}.)\\
\indent In~\cite{Kapur1997,Tiwari2003}, some approaches to constructing the congruence closure of ground equations using completion methods were considered. These approaches capture the efficient techniques from standard term rewriting for congruence closure procedures. In particular, the \emph{abstract congruence closure} approach in~\cite{Tiwari2003} (cf.\;Kapur's approach in~\cite{Kapur1997}) constructs a reduced convergent ground rewrite system $R_S$ for a finite set of ground equations $S$, which consists of either rewrite rules of the form $a\rightarrow c$ or $f(c_1, \ldots, c_{n})\rightarrow c$ or $c\rightarrow d$ for fresh constants $c_1, \ldots, c_{n},c,d$. Furthermore, $R_S$ is a conservative extension of the equational theory induced by $S$ (i.e.\;the congruence closure $CC(S)$) on ground terms, and two ground terms are congruent in $CC(S)$ iff they have the same normal form w.r.t.\;$R_S$. Note that this approach does not require a total termination ordering on ground terms.\\
\indent Congruence closure procedures were extended to congruence closure procedures modulo theories in order to handle interpreted function symbols in the signature~\cite{Bachmair2000,Kapur2019,Baader2020}. The notion of congruence closure modulo associative and commutative (AC) theories was discussed in~\cite{Bachmair2000, Kapur2021}, and the notion of conditional congruence closure with uninterpreted and some interpreted function symbols was considered in~\cite{Kapur2019}.\\
\indent Meanwhile, an equation is a \emph{permutation equation}\cite{Avenhaus2004} if it is of the form $f(x_1,\ldots,x_n)\approx f(x_{\pi(1)},\ldots,\\x_{\pi(n)})$, where $\pi$ is a permutation on the set $\{1,\ldots, n\}$. Commutativity is the simplest case of  permutation equations. Permutation equations are difficult to handle in equational reasoning without using the modulo approach. For example, an ordered completion procedure for \emph{ordered rewriting}~\cite{Bachmair1991c} produces every equation of the form $f(x_1, x_2, \ldots, x_n) \approx f(x_{\rho(1)}, x_{\rho(2)}, \ldots, x_{\rho(n)})$ (up to variable renaming) from two permutation equations $f(x_1, x_2, \ldots, x_n) \approx f(x_2, x_1, x_3, \ldots, x_n)$ and $f(x_1, x_2,\ldots, x_n) \approx f(x_2, x_3, \ldots, x_n, x_1)$, where $\rho$ is a permutation in the symmetric group $S_n$ of cardinality $n!$. (Recall that the symmetric group $S_n$  can be generated by two cycles $(1\,2)$ and $(1\,2\,\cdots\,n)$.) Therefore, it is natural to view permutation equations as ``structural axioms'' (defining a congruence relation on terms) rather than viewing them as  ``simplifiers'' (defining a reduction relation on terms)~\cite{Bachmair1991c}.\\ 
\indent In this paper, we present a framework for generating congruence closure modulo a finite set of permutation equations. To our knowledge, it has not been discussed in the literature, and a polynomial time decision procedure for the word problem for a finite set of ground equations with a fixed set of permutation function symbols has not yet been known.\\ 
\indent Our framework is based on the notion of abstract congruence closure that is particularly useful for term representation and checking $E$-equality between two flat terms for a given set of permutation equations $E$, which does not require an $E$-compatible ordering (cf.\;\cite{Kim2021}). In addition, it also handles  function symbols satisfying each of the following properties: idempotency ($I$), nilpotency ($N$), unit ($U$), $I\cup U$, or $N\cup U$. (If a function symbol is a permutation function symbol satisfying one of the above properties, then it should be a commutative function symbol.)\\
\indent We show that congruence closure modulo a given finite set of permutation equations (with or without the function symbols satisfying the above properties) can be constructed in polynomial time, which provides a polynomial time decision procedure for the word problem for a finite set of ground equations with a fixed set of permutation function symbols (appearing in $E$).


\section{Preliminaries}
We use the standard terminology and definitions of term rewriting~\cite{Baader1998, Dershowitz2001}, congruence closure~\cite{Tiwari2003,Downey1980,Nelson1980}, and permutation groups~\cite{Hungerford1980}. We also use some terminology and the results of permutation equations found in~\cite{Avenhaus2004,Avenhaus2001}.\\
\indent We denote by $\mathcal{T}(\mathcal{F}, \mathcal{X})$ the set of terms over a finite set of function symbols $\mathcal{F}$ and a denumerable set of variables $\mathcal{X}$. We denote by $T(\mathcal{F})$ the set of ground terms over $\mathcal{F}$. We assume that each function symbol in $\mathcal{F}$ has a fixed arity.\\
\indent An \emph{equation} is an expression $s\approx t$, where $s$ and $t$ are (first-order) terms built from $\mathcal{F}$ and $\mathcal{X}$. A \emph{ground equation} (resp.\;\emph{ground term}) is an equation (resp.\;a term) which does not contain any variable.\\
\indent We write $s[u]$ if $u$ is a subterm of $s$ and denote by $s[t]_p$ the term that is obtained from $s$ by replacing the subterm at position $p$ of $s$ by $t$.\\
\indent An \emph{equivalence} is a reflexive, transitive, and symmetric binary relation. An equivalence $\sim$ on terms is a \emph{congruence} if $s\sim t$ implies $u[s]_p\sim u[t]_p$ for all terms $s$, $t$, $u$ and positions $p$.\\
\indent An \emph{equational theory} is a set of equations. We denote by $\approx_{E}$ (called the \emph{equational theory} induced by $E$) the least congruence on $T(\mathcal{F}, \mathcal{X})$ that is stable under substitutions and contains a set of equations $E$. If $s\approx_E t$ for two terms $s$ and $t$, then $s$ and $t$ are $E$-\emph{equivalent}.\\
\indent Given a finite set $S=\{a_i \approx b_i\,|\,1\leq i \leq m\}$ of ground equations where $a_i, b_i \in T(\mathcal{F})$, the \emph{congruence closure} $CC(S)$~\cite{Kapur2019,Baader2020} is the smallest subset of $T(\mathcal{F})\times T(\mathcal{F})$ that contains $S$ and is closed under the following rules: (i) $S\subseteq CC(S)$, (ii) for every $a\in  T(\mathcal{F})$, $a\approx a \in CC(S)$ (\emph{reflexivity}), (iii) if $a\approx b \in CC(S)$, then $b\approx a \in CC(S)$ (\emph{symmetry}), (iv) if $a\approx b$ and $b\approx c \in CC(S)$, then $a\approx c \in CC(S)$ (\emph{transitivity}), and (v) if $f\in \mathcal{F}$ is an $n$-ary function symbol ($n>0$) and $a_1 \approx b_1,\ldots, a_n \approx b_n \in CC(S)$, then $f(a_1, \ldots, a_n) \approx f(b_1, \ldots, b_n)\in CC(S)$ (\emph{monotonicity}). Note that $CC(S)$ is also the equational theory induced by $S$.\\
\indent A (strict) ordering $\succ$ on terms is an irreflexive and transitive relation on $T(\mathcal{F}, \mathcal{X})$.\\
\indent Given a rewrite system $R$ and a set of equations $E$, the rewrite relation $\rightarrow_{R,E}$ on $T(\mathcal{F}, \mathcal{X})$ is defined by $s\rightarrow_{R,E} t$ if there is a non-variable position $p$ in $s$, a rewrite rule $l \rightarrow r \in R$, and a substitution $\sigma$ such that $s|_p  \, {\approx}_E \, l\sigma$ and $t=s[r\sigma]_p$. 
The transitive and reflexive closure of  $\rightarrow_{R,E}$ is denoted by $\xrightarrow{*}_{R,E}$. 
We say that a term $t$ is an $R,E$-\emph{normal form} if there is no term $t^\prime$ such that $t \rightarrow_{R,E} t^\prime$.\\ 
\indent The rewrite relation $\rightarrow_{R/E}$ on $T(\mathcal{F}, \mathcal{X})$ is defined by $s\rightarrow_{R/E} t$ if there are terms $u$ and $v$ such that $s\approx_E u$, $u\rightarrow_R v$, and $v\approx_E t$. We simply say the rewrite relation $\rightarrow_{R/E}$ (resp.\;$\rightarrow_{R,E}$) on $T(\mathcal{F}, \mathcal{X})$ as the rewrite relation $R/E$ (resp.\;$R,E$).
\\
\indent The rewrite relation $R,E$ is \emph{Church-Rosser modulo E} if for all terms $s$ and $t$ with $s \xleftrightarrow{*}_{R\cup E}t$, there are terms $u$ and $v$ such that $s \xrightarrow{*}_{R,E} u \xleftrightarrow{*}_{E} v \xleftarrow{*}_{R,E} t$. The rewrite relation $R,E$ is \emph{convergent modulo E} if $R,E$ is Church-Rosser modulo~$E$ and $R/E$ is well-founded.\\
\indent The \emph{depth} of a term $t$ is defined as $depth(t) = 0$ if $t$ is a variable or a constant and $depth(f(s_1,\ldots,s_n))\\ = 1 + \text{max}\{depth(s_i)\,|\,1\leq i\leq n\}$. A term $t$ is \emph{flat} if its depth is 0 or 1.\\
\indent An equation of the form $f(x_1,\ldots,x_n)=f(x_{\rho(1)},\ldots,x_{\rho(n)})$ is a \emph{permutation equation}\cite{Avenhaus2004} if $\rho$ is a permutation on $\{1,\ldots,n\}$. 
We use  variable naming in such a way that the left-hand side of each equation in a set of permutation equations with the same function symbol has the same name of variables $x_1,\ldots,x_k$ from left to right. (In this paper, we assume that the set of function symbols $\mathcal{F}$ in $T(\mathcal{F}, \mathcal{X})$ is finite and each function symbol in $\mathcal{F}$ has a fixed arity.) 
\\
\indent We denote by $\mathcal{F}_E$ the set of all function symbols occurring in a finite set of permutation equations $E$.\\
\indent If $e:=f(x_1,\ldots,x_n) \approx f(x_{\rho(1)},\ldots,x_{\rho(n)})$ is a permutation equation, then $\rho$ is the permutation of this equation. We denote by $\pi[e]$ the permutation of $e$. For example, $\rho$ is the permutation of the permutation equation $e^\prime:=f(x_1,x_2,x_3, x_4) \approx f(x_1,x_3,x_2, x_4)$ (i.e.\;$\pi[e^\prime]=\rho$) with $\rho(1) =1, \rho(2)=3$, $\rho(3)=2$, and $\rho(4)=4$. Let $E$ be a set of permutation equations with the same top function symbol. Then $\Pi[E]$ is defined as $\Pi[E]:=\{\pi[e]\,|\,e\in E\}$. The permutation group generated by $\Pi[E]$ is denoted by $\mathbin{{<}{\Pi[E]}{>}}$.

\begin{thm}(see Theorem 1.4 in~\cite{Avenhaus2001})\label{thm:congpermGroup} Let $E$ be a set of permutation equations and let $e$ be a permutation equation such that every equation in $E\cup \{e\}$ has the same (top) function symbol. Then $E\models e$ if and only if $\pi[e] \in \mathbin{{<}{\Pi[E]}{>}}$.
\end{thm}

Let $i_1, i_2,\ldots, i_r\,(r \leq n)$ be distinct elements of $I_n = \{1, 2,\ldots, n\}$. Then $(i_1\,i_2\cdots i_r)$, called \emph{a cycle of length} $r$, is defined as the permutation that maps $i_1\mapsto i_2$, $i_2\mapsto i_3$,\ldots, $i_{r-1}$$ \mapsto i_r$ and $i_r\mapsto i_1$, and every other element of $I_n$ maps onto itself. The symmetric group $S_n$ of cardinality $n!$ can be generated by two cycles $(1\,2)$ and $(1\,2\,\cdots\,n)$.

\begin{exam}\normalfont\label{ex:1}
Let $E=\{f(x_1, x_2, x_3, x_4, x_5) \approx f(x_2, x_1, x_3, x_4, x_5), f(x_1, x_2, x_3, x_4, x_5) \approx f(x_2, x_3, x_4,x_5, x_1)\}$. Then $\Pi[E]$ consists of two cycles $\{(1\,2), (1\,2\,3\,4\,5)\}$. Since two cycles $(1\,2)$ and $(1\,2\,3\,4\,5)$ generate the symmetric group $S_5$, we see that $\mathbin{{<}{\Pi[E]}{>}}$ is $S_5$. Therefore, $f(x_1, \ldots, x_5) \approx_E f(x_{\tau(1)},\ldots,x_{\tau(5)})$ for any permutation $\tau \in S_5$ by Theorem~\ref{thm:congpermGroup}.
\end{exam}

\indent Let $E$ be a finite set of permutation equations. Then $E$ can be uniquely decomposed as $\bigcup_{i=1}^n E_i$ such that (i) each $E_i$ is a finite set of permutation equations, and (ii) $E_j$ and $E_k$ with $j\neq k$ are disjoint  such that if $s_j \approx t_j \in E_j$ and $s_k \approx t_k \in E_k$, then $s_j$ and $s_k$ do not have the same top symbol (and are not variants of each other). Since we assume that each function symbol has a fixed arity, each distinct function symbol occurring in $E$ corresponds to a distinct $E_i$ in $E$. We denote by $Eq(f)$ the corresponding equational theory with terms headed by such a function symbol $f$. Now, we may apply Theorem~\ref{thm:congpermGroup} for each  $Eq(f)$ in $E$ with $f\in \mathcal{F}_E$.

\section{Congruence closure modulo permutation equations}
\begin{defn}\label{defn:rules}\normalfont
Let $K$ be a set of constants disjoint from $\mathcal{F}$.
\begin{enumerate}[label=(\roman*)]
\item A $D$-\emph{rule} (w.r.t.\;$\mathcal{F}$ and $K$) is a rewrite rule of the form $f(c_1,\ldots,c_n) \rightarrow c$, where $c_1,\ldots,c_n, c$ are constants in $K$ and $f\in \mathcal{F}$ is an $n$-ary function symbol.
\item A $C$-\emph{rule} (w.r.t.\;$K$) is a rule $c\rightarrow d$, where $c$ and $d$ are constants in $K$.
\end{enumerate}
\end{defn}

In Definition~\ref{defn:rules}(i), note that $f\in \mathcal{F}$ can also be a 0-ary function symbol (i.e. a constant). 
\begin{exam}
Let $E=\{f(x_1, x_2) \approx f(x_2, x_1), g(x_1, x_2, x_3) \approx g(x_2,x_1, x_3)\}$. If $\mathcal{F}=\{a,b,h, f,g\}$ with $\mathcal{F}_E=\{f,g\}$ and $P=\{f(b, g(b, a, a))\approx h(a)\}$, then $D_0=\{a\rightarrow c_0, b\rightarrow c_1, g(c_1, c_0, c_0)\rightarrow c_2, f(c_1, c_2)\rightarrow c_3, h(c_0)\rightarrow c_4\}$ is a possible set of $D$-rules over $\mathcal{F}$, and we have $K=\{c_0, c_1, c_2, c_3, c_4\}$. Using $D_0$, we can simplify the original equations in $P$, which gives the set of $C$ rules, i.e., $C_0=\{c_3 \rightarrow c_4\}$ where $c_3\succ c_4$. 
\end{exam}


\begin{defn}\label{defn:cclosure}\normalfont
Let $E$ be a finite set of permutation equations and $K$ be a set of constants disjoint from $\mathcal{F}$. A ground rewrite system $R=D\cup C$ is a \emph{congruence closure modulo} $E$ (w.r.t.\;$\mathcal{F}$ and $K$) if the following conditions are met:
\begin{enumerate}[label=(\roman*)]
\item $D$ is a set of $D$-rules and $C$ is a set of $C$-rules, and for each constant $c \in K$, there exists at least one ground term $t \in \mathcal{T}(\mathcal{F})$ such that $t\xleftrightarrow{*}_{R,E} c$.
\item $R,E$ is a ground convergent (modulo $E$) rewrite system over $\mathcal{T}(\mathcal{F}\cup K)$.
\end{enumerate}

\indent In addition, given a set of ground equations $P$ over $\mathcal{T}(\mathcal{F}\cup K)$, $R$ is said to be a \emph{congruence closure modulo} $E$ (w.r.t.\;$\mathcal{F}$ and $K$) \emph{for} $P$ if for all ground terms $s$ and $t$ over $\mathcal{T}(\mathcal{F})$,  $s \xleftrightarrow{*}_{P\cup E}t$ iff  there are ground terms $u$ and $v$ over $\mathcal{T}(\mathcal{F}\cup K)$ such that $s \xrightarrow{*}_{R,E} u \xleftrightarrow{*}_{E} v \xleftarrow{*}_{R,E} t$.
\end{defn}

In the following, by $B$-\emph{rules with the interpreted function symbol} $g\in \mathcal{F}$, we mean either the idempotency rule ($I$): $\{g(x,x)\rightarrow x\}$ or the nilpotency rule ($N$): $\{g(x,x) \rightarrow 0\}$ or the unit rule ($U$): $\{g(x, 0) \rightarrow x, g(0, x) \rightarrow x\}$ or $I\cup U$ or $N\cup U$. 
\begin{defn}\label{defn:ccclosure}\normalfont
Let $E$ be a finite set of permutation equations and $K$ be a set of constants disjoint from $\mathcal{F}$. A ground rewrite system $R=D\cup C$ is a \emph{congruence closure modulo} $E\cup B$ (w.r.t.\;$\mathcal{F}$ and $K$) if the following conditions are met:
\begin{enumerate}[label=(\roman*)]
\item $B$ is a set of $B$-rules with the interpreted function symbol $g \in \mathcal{F}$.\footnote{If $g \in \mathcal{F}_E$, then $g$ is a commutative function symbol, i.e., $g(x_1, x_2) \approx g(x_2, x_1) \in E$.}
\item\label{represent} $D$ is a set of $D$-rules and $C$ is a set of $C$-rules, and for each constant $c \in K$, there exists at least one ground term $t \in \mathcal{T}(\mathcal{F})$ such that $t\xleftrightarrow{*}_{R,E} c$.
\item $R\cup B,E$ is a convergent (modulo $E$) rewrite system over $\mathcal{T}(\mathcal{F}\cup K, \mathcal{X})$.\footnote{In this paper, $R\cup B, E$ (resp.\;$R\cup B/E$) denotes $(R\cup B), E$ (resp.\;$(R\cup B)/E$).}
\end{enumerate}

\indent In addition, given a set of ground equations $P$ over $\mathcal{T}(\mathcal{F}\cup K)$, $R$ is said to be a \emph{congruence closure modulo} $E \cup B$ (w.r.t.\;$\mathcal{F}$ and $K$) \emph{for} $P$ if for all ground terms $s$ and $t$ over  $\mathcal{T}(\mathcal{F})$, $s \xleftrightarrow{*}_{P\cup B\cup E}t$ iff  there are ground terms $u$ and $v$ over $\mathcal{T}(\mathcal{F}\cup K)$ such that $s \xrightarrow{*}_{R\cup B,E} u \xleftrightarrow{*}_{E} v \xleftarrow{*}_{R\cup B,E} t$. 
\end{defn}

Note that $B$ or $E$ can be empty in Definition~\ref{defn:ccclosure}. If $B$ is empty, then it is the same as Definition~\ref{defn:cclosure}. 
Also, condition~\ref{represent} in Definition~\ref{defn:ccclosure} states that each constant $c$ in $K$ represents some term in $\mathcal{T}(\mathcal{F})$ w.r.t.\;$R,E$, meaning that $K$ contains no superfluous constants (cf.~\cite{Tiwari2003}).

\begin{defn}\label{defn:cordering}\normalfont
We denote by $W$  the infinite set of constants $\{c_0, c_1, \ldots\}$ such that $W$ is disjoint from $\mathcal{F}$, and denote by $K$ a finite subset chosen from $W$. We define orderings $\succ_K$ on $K$, and $\succ$ and $\succ_{lpo}$ on $\mathcal{T}(\mathcal{F}\cup K)$ as follows:
\begin{enumerate}[label=(\roman*)]
\item $c_i \succ_K c_j$  if $i < j$ for all $c_i, c_j \in K$.
\item $c \succ d$ if $c \succ_{K} d$, and $t\succ c$ if $t \rightarrow c$ is a $D$-rule.
\item $\succ_{lpo}$ is a lexicographic path ordering on $\mathcal{T}(\mathcal{F}\cup K)$, which can be defined from the following assumptions:\\ (iii.1) $c \succ_{lpo} d$ if $c \succ_{K} d$, \\
(iii.2) $t \succ_{lpo} c$ if $t$ is any term headed by a function symbol $f$ in $\mathcal{F}$ and $c$ is any constant in $K$, and\\(iii.3) there is a total precedence on symbols in $\mathcal{F}$.
\end{enumerate}
\end{defn}

Observe that $\succ_{lpo}$ extends $\succ$, and is total on $\mathcal{T}(\mathcal{F}\cup K)$. (If the precedence on $\mathcal{F} \cup K$  is total, then the associated lexicographic path ordering $\succ_{lpo}$ is total on $\mathcal{T}(\mathcal{F}\cup K)$ (see~\cite{Dershowitz2001}).) We emphasize that a partial ordering $\succ$ on $\mathcal{T}(\mathcal{F}\cup K)$ suffices for inference rules in Figure~\ref{fig:fig1}.

\begin{figure}
    \begin{framed}
\begin{tabular}{l}
\LeftLabel {EXTEND:\;\;\;\;\;}
\AxiomC{$(K, P[t], R)$}
\UnaryInfC{$(K\cup\{c\}, P[c], R\cup \{t\rightarrow c\})$}
\DisplayProof\\\\
\quad\quad\quad\quad\quad\quad\; if $t\rightarrow c$ is a $D$-rule, $c \in W-K$, and $t$ occurs in some\\
\quad\quad\quad\quad\quad\quad\;  equation in $P$.\\\\

\LeftLabel  {SIMPLIFY:\;\;\;}
\AxiomC{$(K, P[t], R\cup \{t\rightarrow c\})$}
\UnaryInfC{$(K, P[c], R\cup \{t\rightarrow c\})$}
\DisplayProof\\\\
\quad\quad\quad\quad\quad\quad\; if $t$ occurs in some equation in $P$.\footnote{hi}\\\\

\LeftLabel  {REWRITE:\;\;\;}
\AxiomC{$(K, P, R \cup \{l^\prime \rightarrow r^\prime\})$}
\UnaryInfC{$(K, P\cup \{r\sigma \approx r^\prime\}, R)$}
\DisplayProof\\\\
\quad\quad\quad\quad\quad\quad\; if $l^\prime = l\sigma$, where $l\rightarrow r \in B$.\\\\

\LeftLabel  {ORIENT:\;\;\;\;\;\;}
\AxiomC{$(K, P\cup\{s\approx t\}, R)$}
\UnaryInfC{$(K, P, R\cup\{s \rightarrow t\})$}
\DisplayProof\\\\
\quad\quad\quad\quad\quad\quad\; if $s\succ t$, and $s\rightarrow t$ is a $D$-rule or a $C$-rule.\\\\

\LeftLabel  {DEDUCE:\;\;\;\;\;}
\AxiomC{$(K, P, R\cup \{s\rightarrow c\,,\, t \rightarrow d\})$}
\UnaryInfC{$(K, P\cup \{c\approx d\}, R\cup \{t \rightarrow d\})$}
\DisplayProof\\\\
\quad\quad\quad\quad\quad\quad\; if $s\approx_E t$. \\\\

\LeftLabel  {DELETE:\;\;\;\;\;\;}
\AxiomC{$(K, P\cup\{s\approx t\}, R)$}
\UnaryInfC{$(K, P, R)$}
\DisplayProof\\\\
\quad\quad\quad\quad\quad\quad\; if $s\approx_E t$.\\\\

\LeftLabel  {COMPOSE:\;\;\;}
\AxiomC{$(K, P, R\cup \{t\rightarrow c\,,\, c \rightarrow d\})$}
\UnaryInfC{$(K, P, R\cup \{t\rightarrow d\,,\, c \rightarrow d\})$}
\DisplayProof\\\\

\LeftLabel  {COLLAPSE:\;\;}
\AxiomC{$(K, P, R\cup \{t[c]\rightarrow c^\prime\,,\, c \rightarrow d\})$}
\UnaryInfC{$(K, P, R\cup \{t[d]\rightarrow c^\prime\,,\, c \rightarrow d\})$}
\DisplayProof\\\\
\quad\quad\quad\quad\quad\quad\; if $c$ is a proper subterm of $t$ and $c \rightarrow d$ is a $C$-rule.

\end{tabular}
\end{framed} 
\caption{Inference rules for congruence closure modulo permutation equations}
    \label{fig:fig1}
\end{figure}

Figure~\ref{fig:fig1} shows the inference rules for congruence closure modulo permutation equations, which extends the inference rules for the abstract congruence closure framework in~\cite{Tiwari2003}. We have the additional inference rule called the REWRITE rule in Figure~\ref{fig:fig1}. Also, we use the $E$-equality $\approx_E$ instead of the equality $\approx$ for the DEDUCE and DELETE inference rules. We write $(K, P, R) \vdash (K^\prime, P^\prime, R^\prime)$ to indicate that $(K^\prime, P^\prime, R^\prime)$ can be obtained from $(K, P, R)$ by application of an inference rule in Figure~\ref{fig:fig1}, where $K$ denotes a set of new constants (see Definition~\ref{defn:cordering}), $P$ a set of equations, and $R$ a set of rewrite rules consisting of $C$-rules and $D$-rules. Also,  in Figure~\ref{fig:fig1}, $B$ denotes a set of $B$-rules. A \emph{derivation} is a sequence of states $(K_0, P_0, R_0)\vdash(K_1, P_1, R_1)\vdash\cdots$.

\begin{lem}\label{lem:congequiv}
If $(K, P, R) \vdash (K^\prime, P^\prime, R^\prime)$, then for all $u$ and $v$ in $\mathcal{T}(\mathcal{F}\cup K)$, we have $u \xleftrightarrow{*}_{E\cup B\cup P^\prime\cup R^\prime} v$  if and only if $u \xleftrightarrow{*}_{E\cup B\cup P\cup R} v$.
\end{lem}
\begin{proof}
We consider each application of an inference rule $\tau$ for $(K, P, R) \vdash (K^\prime, P^\prime, R^\prime)$. If $\tau$ is EXTEND, SIMPLIFY, ORIENT, DELETE,  COLLAPSE, or COMPOSE, then the conclusion can be verified similarly to~\cite{Tiwari2003,Bachmair1991c}.\\
\indent If $\tau$ is REWRITE, then we let $P=\bar{P}$, $R=\bar{R} \cup \{l^\prime \rightarrow r^\prime\}$,  $R^\prime=\bar{R}$, $P^\prime=\bar{P} \cup \{r\sigma\approx r^\prime\}$, and $K=K^\prime$. Since $(K\cup P\cup R)-(K^\prime\cup P^\prime \cup R^\prime)=\{l^\prime\rightarrow r^\prime\}$, we need to show that $l^\prime \xleftrightarrow{*}_{E\cup B\cup P^\prime\cup R^\prime}r^\prime$. As $l^\prime=l\sigma\rightarrow_{B}r\sigma \leftrightarrow_{P^\prime}r^\prime$, we have  $l^\prime \xleftrightarrow{*}_{E\cup B\cup P^\prime\cup R^\prime}r^\prime$. Conversely, since $(K^\prime \cup P^\prime \cup R^\prime)-(K\cup P \cup R)=\{r\sigma\approx r^\prime\}$, we need to show that $r\sigma\xleftrightarrow{*}_{E\cup B\cup P\cup R} r^\prime$. As $r\sigma\leftarrow_B l\sigma=l^\prime\rightarrow_R r^\prime$, we have $r\sigma\xleftrightarrow{*}_{E\cup B\cup P\cup R} r^\prime$.\\
\indent If $\tau$ is DEDUCE, then let $R=\bar{R}\cup \{s\rightarrow c, t \rightarrow d\}$, $P^\prime=P\cup \{c\approx d\}$, $R^\prime = \bar{R} \cup \{t\rightarrow d\}$, and $K=K^\prime$, where $s\approx_E t$. Since $(K\cup P\cup R)-(K^\prime\cup P^\prime \cup R^\prime)=\{s\rightarrow c\}$, we need to show that $s \xleftrightarrow{*}_{E\cup B\cup P^\prime\cup R^\prime} c$.  As $s\xleftrightarrow{*}_E t \rightarrow_{R^\prime} d \leftrightarrow_{P^\prime} c$, we have $s \xleftrightarrow{*}_{E\cup B\cup P^\prime\cup R^\prime} c$. Conversely, since $(K^\prime \cup P^\prime \cup R^\prime)-(K\cup P \cup R)=\{c \approx d\}$, we need to show that $c \xleftrightarrow{*}_{E\cup B\cup P\cup R} d$. As $c\leftarrow_R s \xleftrightarrow{*}_E t \rightarrow_R d$, we have $c \xleftrightarrow{*}_{E\cup B\cup P\cup R} d$.
\end{proof}

\begin{defn}\normalfont
(i) A derivation is said to be \emph{fair} if any inference rule that is continuously enabled is applied eventually.\\
(ii) By a \emph{fair} $\mu$-\emph{derivation}, we mean that the EXTEND and SIMPLIFY rule are applied eagerly in a fair derivation.
\end{defn}

\begin{thm}
Let $(K_0, P_0, R_0) \vdash (K_1, P_1, R_1)\vdash\cdots$ be a fair $\mu$-derivation such that $P_0$ is a finite set of ground equations with $K_0=\emptyset$ and $R_0=\emptyset$. \\
(i) Each fair $\mu$-derivation starting from the initial state $(K_0, P_0, R_0)$ is finite.\\
(ii) If $(K_n, P_n, R_n)$ is a final state $($i.e.\;no inference rule can be applied to $(K_n, P_n, R_n))$ of a fair $\mu$-derivation starting from the initial state $(K_0, P_0, R_0)$, then $R_n\cup B, E$ is convergent modulo $E$, and $R_n$ is a congruence closure modulo $E\cup B$ for $P_0$.
\end{thm}
\begin{proof}
Since $(K_0, P_0, R_0) \vdash (K_1, P_1, R_1)\vdash\cdots$ is a fair $\mu$-derivation, this derivation can be written as  $(K_0, P_0, R_0) \vdash^* (K_m, P_m, R_m)\vdash (K_{m+1}, E_{m+1}, R_{m+1}) \vdash \cdots$, where the derivation $(K_m, P_m, R_m)\vdash (K_{m+1}, E_{m+1},\\ R_{m+1}) \vdash \cdots$ does not involve any application of the EXTEND rule, so we have the set $K_m = K_{m+1}=\cdots$.\\
\indent For (i), we provide a more concrete result in the following Lemma~\ref{lem:derivedlength}.\\
\indent  For (ii), let $(K_n, P_n, R_n)$ be a final state of a fair $\mu$-derivation starting from the state $(K_0, P_0, R_0)$. (Note that each fair $\mu$-derivation starting from the initial state $(K_0, P_0, R_0)$ is finite by (i), so we have some final state.) Observe that $P_m, P_{m+1}, \ldots$ either contains only $C$-equations or is empty, and $\succ$ can orient those $C$ equations, so $P_n=\emptyset$. Since $l\succ_{lpo} r$ for all rules $l\rightarrow r \in R_n\cup B$, we see that $R_n\cup B/E$ is terminating. 

Also, since $R_n\cup B$ is non-overlapping w.r.t.\;the rewrite system $R_n\cup B, E$ (i.e.\;there is no non-trivial critical pair between rules in  $R_n\cup B$), $R_n \cup B, E$ is Church-Rosser modulo $E$ by the critical pair lemma~\cite{Bachmair1991c}. Thus, $R_n\cup B, E$ is convergent modulo $E$. \\
\indent Finally, we show that for each constant $c \in K$, there exists at least one ground term $t \in \mathcal{T}(\mathcal{F})$ such that $t\xleftrightarrow{*}_{R_n,E} c$ by induction. Let $c$ be a constant in $K$ and $f(c_1,\ldots, c_k) \rightarrow c$ be the corresponding extension rule for $c$ when $c$ was added. By induction hypothesis, we have $s_i \xleftrightarrow{*}_{R_n, E} c_i$, and thus $f(s_1,\ldots, s_k) \xleftrightarrow{*}_{R_n, E} f(c_1,\ldots, c_k) \rightarrow_{\cup_i R_i} c$. By Lemma~\ref{lem:congequiv}, we also have $f(s_1,\ldots, s_k) \xleftrightarrow{*}_{R_n\cup P_n B\cup E} c$, and thus $f(s_1,\ldots, s_k) \xleftrightarrow{*}_{R_n\cup  B\cup E} c$ because $P_n = \emptyset$. As $R_n\cup B, E$ is convergent modulo $E$ and no more REWRITE rule can be applied to $f(s_1,\ldots, s_k)$ by fairness of the derivation, we have $f(s_1,\ldots, s_k) \xleftrightarrow{*}_{R_n, E} c$. Thus $R_n$ is a congruence closure modulo $E\cup B$ for $P_0$.
\end{proof}

In the following lemma, recall that function symbols in $\mathcal{F}$ include 0-ary function symbols in $\mathcal{F}$, i.e., constants in $\mathcal{F}$.
\begin{lem}\label{lem:derivedlength}
Let $(K_0, P_0, R_0) \vdash (K_1, P_1, R_1)\vdash\cdots$ be a fair $\mu$-derivation such that $P_0$ is a finite set of ground equations with $K_0=\emptyset$ and $R_0=\emptyset$. Then its derivation length is bounded by $O(n^2)$, where $n$ is the sum of the sizes (number of symbols) of the left-hand and right-hand sides of equations in $P_0$.
\end{lem}
\begin{proof}
We show that the number of applications of each rule in Figure~\ref{fig:fig1} in a fair $\mu$-derivation is bounded above by $O(n^2)$, where $n$ is the sum of the sizes (number of symbols in $\mathcal{F}$) of the left-hand and right-hand sides of equations in $P_0$. Since $(K_0, P_0, R_0) \vdash (K_1, P_1, R_1)\vdash\cdots$ is a fair $\mu$-derivation, we may write this derivation as\\

\indent  $(K_0, P_0, R_0) \vdash^* (K_m, P_m, R_m)\vdash (K_{m+1}, E_{m+1}, R_{m+1}) \vdash \cdots$,\\

\noindent where the derivation $(K_m, P_m, R_m)\vdash (K_{m+1}, E_{m+1}, R_{m+1}) \vdash \cdots$ does not involve any application of the EXTEND rule, and thus we have the finite set $K_m = K_{m+1}=\cdots$.

\begin{enumerate}[label=(\roman*)]
\item\label{firstitem} The total number of the EXTEND inference steps is bounded by $O(n)$. This is because the total number of $\mathcal{F}$-symbols in the second component of the state is not increasing for each transition step\footnote{The only exception is the case where the REWRITE inference step using  the nilpotency rule introduces constant 0 in the second component of the state, where 0 does not occur in $P_0$. But this requires at most one additional EXTEND inference step.}  and each EXTEND inference step decreases this number by one.
\item A derivation step by the SIMPLIFY, REWRITE, DEDUCE, COMPOSE, or COLLAPSE rule either reduces the number of function symbols of $\mathcal{F}$ in $R_i\cup P_i$ or rewrites some constant. The length of a rewriting sequence $c_1\rightarrow c_2\rightarrow \cdots$ is bounded by $|K_m|$. (Here, $|K_m|$ is $O(n)$ because the total number of the EXTEND inference steps is bounded by $O(n)$ as discussed in~\ref{firstitem}.) Also, the total number of symbols in $P_i \cup R_i$ is bounded by $O(n + |K_m|)$,\footnote{The total number of symbols in $P_i \cup R_i$ for each transition step does not increase except by an EXTEND inference step, where an EXTEND inference step may increase this number by two.} which is also $O(n)$. This means that the total number of the SIMPLIFY, DEDUCE, COLLAPSE, or COMPOSE inference steps is bounded by $O(n^2)$. (Note that rewriting constants takes $O(n^2)$ because there are at most $O(n)$ constants, and the length of a rewriting sequence for each constant is bounded by $O(n)$.)
\item The total number of the DELETE inference steps is bounded by $O(n + |K_m|)$ (i.e. $O(n)$) because the total number of symbols in $P_i \cup R_i$ is bounded by $O(n + |K_m|)$.
\item The total number of the ORIENT inference steps is bounded by the total number of EXTEND, SIMPLIFY, DEDUCE, COLLAPSE, and COMPOSE inference steps, which is $O(n^2)$. Note that each ORIENT inference step neither increases the number of function symbols nor the number of constants.
\end{enumerate}
Thus, the derivation length of any fair $\mu$-derivation starting from $(K_0, P_0, R_0)$ is bounded by $O(n^2)$.
\end{proof}

Given a finite (fixed) set of permutation equations $E$ and two terms $s=f(s_1,\ldots,s_k)$ and $t=f(t_1,\ldots,\\t_k)$ with $f\in \mathcal{F}_E$, we can determine whether $s\approx_E t$ in $O(n^2)$ time (measured in $n=|s| + |t|$) using an additional data structure (i.e.\;a table) that can be constructed in polynomial time~\cite{Avenhaus2004}. If $s$ and $t$ are both flat, then we can determine whether $s\approx_E t$ in $O(n)$ time using the following procedure with a table that can be constructed in polynomial time (see~\cite{Avenhaus2004}).\\\\
\emph{Equality-Test}($s$, $t$)\\
Input: $s=f(c_1,\ldots,c_i)$ and $t=g(d_1,\ldots,d_j)$, where $s$ and $t$ are both flat.\\
Output: If $s\approx_E t$, then return true. Otherwise, return false.
\begin{enumerate}[nosep]
\item Determine whether $s$ and $t$ are headed by the same function symbol (i.e.\;$f=g$ and thus $i=j$). If not, then return false. If it is true, then consider the following:
\item Determine whether $f\in \mathcal{F}_E$. If not, then $s$ and $t$ are compared by syntactic equality, and return true if they are syntactically equal. Otherwise, if $f\in\mathcal{F}_E$, then consider the following:
\item Determine whether $s\approx_E t$ using the \emph{TestEq} procedure in~\cite{Avenhaus2004}.\\
\end{enumerate}

It is easy to see that steps 1 and 2 of the \emph{Equality-Test}($s$, $t$) procedure take at most $O(n)$ time. For step 3, which corresponds to the case $f=g$ and $f\in \mathcal{F}_E$, it takes $O(n)$ time for comparing two multisets. 

If they are equal, then $s$ and $t$ are further compared in constant time using the \emph{TestEq} procedure in~\cite{Avenhaus2004} with a table that can be constructed in polynomial time. (Note that the arity of all $f\in \mathcal{F}_E$ and the size of the data structure (i.e.\;table) is bounded by a constant independent of the size of the input terms.) Therefore, the \emph{Equality}-\emph{Test}($s$, $t$) procedure takes $O(n)$ time using a table that can be constructed in polynomial time. In what follows, we denote by $Table(Eq(f))$ this table for each $f\in \mathcal{F}_E$ for a finite (fixed) set of permutation equations $E$.

\begin{thm}\label{thm:conclosure} Given the table $Table(Eq(f))$ for each $f\in \mathcal{F}_E$, a congruence closure modulo $E\cup B$ for a finite set of ground equations $P$ can be computed in $O(n^3)$ time, where $n$ is the sum of the sizes (number of symbols) of the left and right sides of equations in $P$.
\end{thm}
\begin{proof}
We first construct a fair $\mu$-derivation $(K_0, P_0, R_0) \vdash (K_1, P_1, R_1)\vdash\cdots$  such that $P_0$ is a finite set of ground equations with $K_0=\emptyset$ and $R_0=\emptyset$.\\
\indent It is easy to see that each EXTEND, SIMPLIFY, ORIENT, COMPOSE, and COLLAPSE inference step in the derivation takes $O(n)$ time.\\
\indent Each REWRITE inference step in the derivation takes $O(n)$ time because we only need to consider for rules $g(x,x)\rightarrow x$, $g(x,x) \rightarrow 0$, $g(x, 0) \rightarrow x$, and $g(0, x) \rightarrow x$ for some interpreted function symbol $g\in \mathcal{F}$.\\
\indent Each DEDUCE inference step in the derivation takes $O(n)$ time for checking $E$-equality (see the \emph{Equality-Test}($s$, $t$) procedure) between two left-hand side terms $s$ and $t$ in $R_i$. Similarly, each DELETE inference step in the derivation takes $O(n)$ time for checking $E$-equality.\\
\indent By Lemma~\ref{lem:derivedlength}, we know that the derivation length of a fair $\mu$-derivation is bounded by $O(n^2)$. Since each inference step in the derivation takes $O(n)$ time, a congruence closure modulo $E\cup B$ for a finite set of ground equations $P$ can be computed in $O(n^3)$ time.
\end{proof}

\begin{cor}
The word problem for a finite set of ground equations $P$ with a fixed set of permutation function symbols is decidable in polynomial time.
\end{cor}
\begin{proof}
We can decide whether $s\approx_E^?t$ for two ground terms $s$ and $t$ using a congruence closure modulo $E$ for $P$. By Theorem~\ref{thm:conclosure}, we can compute a congruence closure modulo $E$ for $P$ in polynomial time by constructing and using the table $Table(Eq(f))$ for each $f\in \mathcal{F}_E$. Let $R$ be a congruence closure modulo $E$ for $P$. We obtain each normal form of $s$ and $t$ using $R$. We first rewrite each constant symbol in $\mathcal{F}$ of $s$ and $t$ to a new constant symbol in $K$ obtained from constructing $R$, which takes $O(m)$ time where $m=|s|+|t|$. Each rewrite step either reduces the size of a term or rewrites a constant in $K$ to another constant in $K$.  The length of a rewriting sequence $c_1\rightarrow c_2\rightarrow \cdots$ is bounded by $|K|$ (i.e.\;$O(n)$), where $n$ is the sum of the sizes of the left-hand and right-hand sides of equations in $P$. We may also infer that the sum of the sizes of the left-hand and right-hand sides of the rewrite rules in $R$ is $O(n+|K|)$, which is $O(n)$. Each rewrite step takes at most $O(n^2)$ time using $R$ and the  \emph{Equality}-\emph{Test} procedure. By combining these steps together, we can decide whether $s\approx_E^?t$ for two ground terms $s$ and $t$ using their normal forms in polynomial time.
\end{proof}
The above corollary also holds if some function symbols (not necessarily permutation function symbols) satisfies the properties, such as idempotency ($I$), nilpotency ($N$), unit ($U$), $I\cup U$, or $N\cup U$.

\section{Example of congruence closure modulo $E\cup B$}
Let $B$ be the set of the equation for an idempotency function symbol $g$, i.e., $B=\{g(x,x)\rightarrow x\}$ and let $E$ be the following set of permutation equations:\\\\
$E=\{f(x_1, x_2, x_3, x_4,x_5, x_6, x_7, x_8) \approx  f(x_2, x_1, x_3, x_4,x_5, x_6, x_7, x_8),\\
\indent\;\;\;\;f(x_1, x_2, x_3, x_4,x_5, x_6, x_7, x_8) \approx f(x_2, x_3, x_4, x_1,x_5, x_6, x_7, x_8),\\
\indent\;\;\;\; f(x_1, x_2, x_3, x_4,x_5, x_6, x_7, x_8) \approx f(x_1, x_2, x_3, x_4,x_6, x_5, x_7, x_8),\\
\indent\;\;\;\; f(x_1, x_2, x_3, x_4,x_5, x_6, x_7, x_8) \approx f(x_1, x_2, x_3, x_4,x_5, x_6, x_8, x_7)\}$.\\

In this example, we may view each variable $x_i$ as a switch in a specially designed electric board, where each variable will be assigned to either constant $T$ (representing ``on'') or constant $F$ (representing ``off''). Each ground term $f(c_1, c_2, c_3, c_4,c_5, c_6, c_7, c_8)$ with $c_i = T$ or $F$ represents a certain state of this electric board. There is a special transformation button in this electric board, which may transform one state to another state of the electic board. This transformation button is represented by a function with symbol $h\notin \mathcal{F}_E$. The problem is to determine if a certain state in the electric board (represented by a term) generates a fault state (represented by term $\bot$). We see that $ \prod[E]=\{(1\,2), (1\,2\,3\,4), (5\,6),(7\,8)\}$, which means that $f(x_1, x_2, x_3, x_4,x_5, x_6, x_7, x_8) \approx_E  f(x_{\rho(1)}, x_{\rho(2)}, x_{\rho(3)}, x_{\rho(4)},x_5, x_6, x_7, x_8)$ for any permutation $\rho$ on the set $\{1,2,3,4\}$, $f(x_1, x_2, x_3, x_4,x_5, x_6, x_7, x_8) \approx_E  f(x_1, x_2, x_3, x_4,x_{\pi(5)}, x_{\pi(6)}, x_7, x_8)$ for any permutation $\pi$ on the set $\{5,6\}$, and $f(x_1, x_2, x_3, x_4,x_5, x_6, x_7, x_8) \approx_E  f(x_1, x_2, x_3, x_4,x_5, x_6, x_{\tau(7)}, x_{\tau(8)})$ for any permutation $\tau$ on the set $\{7,8\}$  (see Thereom~\ref{thm:congpermGroup}). Therefore, eight switches in the board are partitioned into three components, i.e.\;$\{x_1, x_2, x_3, x_4\}$, $\{x_5,x_6\}$ and $\{x_7, x_8\}$, where the order of ``switch on'' or ``switch off'' does not matter in each component. For example, $f(T, T, F, F,T, F, T, F)\approx_E f(F, F, T, T,F, T,T, F)$. Meanwhile, $g$ is an idempotent function symbol, which serves as a comparator for fault states. For example, if $g(\bot, f(F, F, F, T,T, T, T, F))$, then it is $\bot$ if $f(F, F, F, T,T, T, T, F)$ is $\bot$. Now we start with the following set of ground equations:
\begin{longtable}{l}
$1. \quad f(T, T, T, T,T, T, T, T)\approx \bot$\\
$2. \quad h(f(F, F, F, F,F, F, F, F)) \approx  f(F, T, F, T,F, T, F, T)$\\
$3. \quad f(T, F, F, F,F, F,F, T)\approx g(\bot, h(f(T, T, T, T,F, T,F, T)))$\\
$4. \quad h(f(T, F, T, F,T, F, T, F)) \approx  f(F, F, F, F,T, T, T, T)$\\
$5. \quad f(F, F, F, F,T, T, T, T) \approx  f(T, T, T, T,F, F, F, F)$\\
$6. \quad h(f(T, T, T, T,F, F, F, F)) \approx  f(T, T, T, T,T, F, T, F)$\\
$7. \quad h(f(T, T, T, T,F, T,F, T)) \approx  f(T, T, T, T,T, T, T, T)$
\end{longtable}
\indent We show that, for example, each of $h^4(f(F, F, F, F,F, F, F, F))$ and $f(T, F, F, F,F, F,F, T)$ is a fault state. (For notational brevity, by $h^i(t)$, we mean the function symbol $h$ is applied to term $h^{i-1}(t)$ with $h^0(t)$ denoting $t$.) The initial state is $(K_0, P_0, R_0)$, where $K_0=R_0=\emptyset$ and $P_0$ consists of the above equations $1-7$. We apply a fair $\mu$-derivation starting with $(K_0, P_0, R_0)$ and some intermediate and repetitive steps are omitted for clarity. In the following, each rewrite rule is an element of some $R_i$ and each equation is an element of some $P_j$.  We assume that $c_i \succ c_j$ if $i < j$.

\begin{longtable}{ll}
\noindent
$1(a). \quad T\rightarrow c_1, F\rightarrow c_2, \bot \rightarrow c_3$ &\quad\quad EXTEND and SIMPLIFY for 1\\
$1(b). \quad f(c_1, c_1, c_1, c_1,c_1, c_1, c_1, c_1) \rightarrow c_4$ & \\
$1(c). \quad c_4 \approx c_3$ & \\
$2(a). \quad f(c_2, c_2, c_2, c_2,c_2,c_2, c_2, c_2)\rightarrow  c_5$ &\quad\quad EXTEND and SIMPLIFY for 2 \\
$2(b). \quad h(c_5)\rightarrow c_6$ &\\
$2(c). \quad  f(c_2, c_1, c_2, c_1,c_2, c_1, c_2, c_1) \rightarrow c_7$ &\\
$2(d). \quad c_6 \approx c_7$ &\\
$3(a). \quad  f(c_1, c_2, c_2, c_2,c_2, c_2,c_2, c_1) \rightarrow c_8$ &\quad\quad EXTEND and SIMPLIFY for 3\\
$3(b). \quad f(c_1, c_1, c_1, c_1,c_2, c_1,c_2, c_1) \rightarrow  c_9$ &\\
$3(c). \quad h(c_9) \rightarrow c_{10}$ &\\
$3(d). \quad g(c_3, c_{10})\rightarrow  c_{11}$ &\\
$3(e). \quad c_8\approx c_{11}$ &\\
$4(a). \quad f(c_1, c_2, c_1, c_2,c_1, c_2, c_1, c_2) \rightarrow c_{12}$ &\quad\quad EXTEND and SIMPLIFY for 4\\
$4(b). \quad h(c_{12}) \rightarrow  c_{13}$ &\\
$4(c). \quad f(c_2, c_2, c_2, c_2,c_1, c_1, c_1, c_1)\rightarrow  c_{14}$ &\\
$4(d). \quad c_{13}\approx c_{14}$ &\\
$5(a). \quad f(c_2, c_2, c_2, c_2,c_1, c_1, c_1, c_1) \rightarrow c_{15}$ &\quad\quad EXTEND and SIMPLIFY for 5\\
$5(b). \quad f(c_1, c_1, c_1, c_1,c_2, c_2, c_2, c_2) \rightarrow c_{16}$ &\\
$5(c). \quad c_{15}\approx c_{16}$ &\\
$6(a). \quad f(c_1, c_1, c_1, c_1,c_2, c_2, c_2, c_2) \rightarrow c_{17}$ &\quad\quad EXTEND and SIMPLIFY for 6\\
$6(b). \quad h(c_{17}) \rightarrow c_{18}$ &\\
$6(c). \quad  f(c_1, c_1, c_1, c_1,c_1, c_2, c_1, c_2)\rightarrow c_{19}$ &\\
$6(d). \quad c_{18}\approx c_{19}$ &\\
$7(a). \quad f(c_1, c_1, c_1, c_1,c_2, c_1,c_2, c_1) \rightarrow c_{20}$ &\quad\quad EXTEND and SIMPLIFY for 7\\
$7(b). \quad h(c_{20}) \rightarrow c_{21}$ &\\
$7(c). \quad  f(c_1, c_1, c_1, c_1,c_1, c_1, c_1, c_1)\rightarrow c_{22}$ &\\
$7(d). \quad c_{21}\approx c_{22}$ &\\
$8(a). \quad c_7 \approx c_{12}$ \;(Rule 2(c) is now removed.)&\quad\quad DEDUCE with 2(c) and 4(a)\\
$8(b). \quad c_{14} \approx c_{15}$ (Rule 4(c) is now removed.)&\quad\quad DEDUCE with 4(c) and 5(a)\\
$8(c). \quad c_{16} \approx c_{17}$ (Rule 5(b) is now removed.)&\quad\quad DEDUCE with 5(b) and 6(a)\\
$8(d). \quad c_{9} \approx c_{20}$ \;(Rule 3(b) is now removed.)&\quad\quad DEDUCE with 3(b) and 7(a)\\
$8(e). \quad c_{19} \approx c_{20}$ (Rule 6(c) is now removed.)&\quad\quad DEDUCE with 6(c) and 7(a)\\
$8(f). \quad c_{4} \approx c_{22}$ \;(Rule 1(b) is now removed.)&\quad\quad DEDUCE with 1(b) and 7(c)\\
\end{longtable}

We next orient equations into $C$-rules and apply other inference rules. The set of $C$-rules is $C=\{c_3\rightarrow c_4, c_6 \rightarrow c_7, c_8 \rightarrow c_{11}, c_{13} \rightarrow c_{14}, c_{15}\rightarrow c_{16}, c_{18} \rightarrow c_{19}, c_{21}\rightarrow c_{22}, c_7 \rightarrow c_{12}, c_{14}\rightarrow c_{15}, c_{16} \rightarrow c_{17}, c_9 \rightarrow c_{20}, c_{19} \rightarrow c_{20}, c_4 \rightarrow c_{22}\}$. Using DEDUCE, COMPOSE, and ORIENT inference steps, it becomes $C^\prime=\{c_3 \rightarrow c_{22}, c_6 \rightarrow c_{12}, c_8 \rightarrow c_{11}, c_{13} \rightarrow c_{17}, c_{15}\rightarrow c_{17}, c_{18} \rightarrow c_{20}, c_{21}\rightarrow c_{22}, c_7 \rightarrow c_{12}, c_{14}\rightarrow c_{17}, c_{16} \rightarrow c_{17}, c_9 \rightarrow c_{20}, c_{19} \rightarrow c_{20}, c_4 \rightarrow c_{22}\}$. The REWRITE inference step 9(d) is available after the following inference steps $9(a)$, $9(b)$, and $9(c)$:
\begin{longtable}{ll}
\noindent
$9(a). \quad h(c_{20}) \rightarrow c_{10}$ &\quad\quad COLLAPSE 3(c) with $c_9 \rightarrow c_{20}$\\
$9(b). \quad c_{10} \rightarrow c_{22}$ &\quad\quad DEDUCE with 7(b) and 9(a), ORIENT, COMPOSE\\
$9(c). \quad g(c_{22}, c_{22})\rightarrow c_{11}$ &\quad\quad COLLAPSE 3(d) with $c_3 \rightarrow c_{22}$ and $c_{10} \rightarrow c_{22}$\\
$9(d). \quad c_{11} \rightarrow c_{22}$ &\quad\quad REWRITE 9(c), ORIENT\\
\end{longtable}

In the above, Rule $3(c)$ is removed after $9(a)$, Rule $9(a)$ is removed after $9(b)$, Rule $3(d)$ is removed after $9(c)$, and Rule $9(c)$ is removed after $9(d)$. We may obtain a congruence closure $R_n=C_n \cup D_n$ modulo $E\cup B$ for $P_0$ for some $n$ with some additional inference steps, where $C_n=\{c_3 \rightarrow c_{22}, c_4 \rightarrow c_{22}, c_6 \rightarrow c_{12}, c_7 \rightarrow c_{12}, c_8 \rightarrow c_{11},  c_9 \rightarrow c_{20},c_{10} \rightarrow c_{22},c_{11} \rightarrow c_{22},c_{13} \rightarrow c_{17},c_{14}\rightarrow c_{17}, c_{15}\rightarrow c_{17},c_{16} \rightarrow c_{17}, c_{18} \rightarrow c_{20},c_{21}\rightarrow c_{22}\}$ and $D_n$ consists of the following set of rules:

\begin{longtable}{ll}
\noindent
D1: $T\rightarrow c_1$ & \quad D2: $F\rightarrow c_2$\\
D3: $\bot \rightarrow c_{22}$ &\quad D4: $f(c_2, c_2, c_2, c_2,c_2,c_2, c_2, c_2)\rightarrow  c_5$\\
D5: $h(c_5)\rightarrow c_{12}$ &\quad D6: $f(c_1, c_2, c_2, c_2,c_2, c_2,c_2, c_1) \rightarrow c_{22}$\\
D7: $f(c_1, c_2, c_1, c_2,c_1, c_2, c_1, c_2) \rightarrow c_{12}$ &\quad D8: $h(c_{12}) \rightarrow  c_{17}$\\
D9: $f(c_2, c_2, c_2, c_2,c_1, c_1, c_1, c_1) \rightarrow c_{17}$  &\quad D10: $f(c_1, c_1, c_1, c_1,c_2, c_2, c_2, c_2) \rightarrow c_{17}$\\
D11: $h(c_{17}) \rightarrow c_{20}$ &\quad D12: $f(c_1, c_1, c_1, c_1,c_2, c_1,c_2, c_1) \rightarrow c_{20}$\\
D13: $h(c_{20}) \rightarrow c_{22}$ &\quad D14: $f(c_1, c_1, c_1, c_1,c_1, c_1, c_1, c_1)\rightarrow c_{22}$
\end{longtable}

\indent Now we determine whether $h^4(f(F, F, F, F,F, F, F, F))$ is a fault state: i.e., $h^4(f(F, F, F, F,F,F, F, F)) \\\approx_{R_n \cup B \cup E}^? \bot$. Since $h^4(f(F, F, F, F,F, F, F, F))\xrightarrow{*}_{R_n,E}h^4(f(c_2, c_2, c_2, c_2,c_2,c_2, c_2, c_2))\rightarrow$$_{R_n,E}\;h^4(c_5)\rightarrow_{R_n,E}h^3(c_{12})\rightarrow_{R_n,E}h^2(c_{17})\rightarrow_{R_n,E}h(c_{20})\rightarrow_{R_n,E}c_{22}$ and $\bot\rightarrow_{R_n,E}c_{22}$, it is a fault state. Similarly, we can determine whether $f(T, F, F, F,F, F,F, T)$ is a fault state. Since $f(T, F, F, F,F, F,F,T)\xrightarrow{*}_{R_n,E}f(c_1,c_2, c_2, c_2,c_2,\\ c_2,c_2, c_1)\rightarrow_{R_n,E}c_{11}\rightarrow_{R_n,E}c_{22}$ and $\bot\rightarrow_{R_n,E}c_{22}$, it is a fault state. Meanwhile, $h^2(f(F, F, T, T,F, T, F, T))$ is not a fault state, i.e., $h^2(f(F, F, T, T,F,T, F, T))\not\approx_{R_n \cup B \cup E} \bot$. Since $h^2(f(F, F, T, T, F, T, F, T))\xrightarrow{*}_{R_n\cup B,E}h^2(f(c_2, c_2, c_1, c_1, c_2, c_1, c_2, c_1))\rightarrow_{R_n\cup B, E} h^2(c_{12})\rightarrow_{R_n\cup B,E}h(c_{17})\rightarrow_{R_n\cup B,E}c_{20}$ and $\bot\rightarrow_{R_n\cup B,E}c_{22}$, it is not a fault state.\\

\section{Conclusion}
We have presented a framework for constructing congruence closure modulo a finite set of permutation equations $E$, extending the abstract congruence closure framework for handling permutation function symbols with or without the interpreted function symbols (not necessarily permutation function symbols) satisfying each of the following properties: idempotency ($I$), nilpotency ($N$), unit ($U$), $I\cup U$, or $N\cup U$. 
We have provided a polynomial time decision procedure for the word problem for a finite set of ground equations with a fixed set of permutation function symbols by constructing congruence closure modulo $E$.\\
\indent Although congruence closure procedures have been widely used in software/hardware verfication and satisfiability modulo theories (SMT) solvers, congruence closure procedures with built-in permutations have not been well studied. We believe that our framework for constructing congruence closure modulo permutation equations has practical significance to software/hardware verfication and SMT solvers involving built-in permutations, where built-in permutations are represented by a finite set of permutation equations containing permutation function symbols.
\bibliographystyle{eptcs}
\bibliography{references}
\end{document}